\documentclass[11pt,letterpaper]{article}
\usepackage[margin=1in]{geometry}
\usepackage[T1]{fontenc}
\usepackage{lmodern}
\usepackage{amsmath,amssymb,amsthm,mathtools}
\usepackage{enumitem,microtype}
\usepackage[colorlinks=true,linkcolor=blue,citecolor=blue,urlcolor=blue]{hyperref}
\newcommand{\doi}[1]{\href{https://doi.org/#1}{doi:\,\nolinkurl{#1}}}
\newtheorem{theorem}{Theorem}[section]
\newtheorem{lemma}[theorem]{Lemma}
\newtheorem{proposition}[theorem]{Proposition}
\newtheorem{corollary}[theorem]{Corollary}
\newcommand{\R}{\mathbb R}\newcommand{\Q}{\mathbb Q}\newcommand{\Z}{\mathbb Z}
\newcommand{\norm}[1]{\lVert #1\rVert}
\newcommand{\supp}{\operatorname{supp}}
\newcommand{\im}{\operatorname{im}}
\newcommand{\poly}{\operatorname{poly}}
\newcommand{\enc}[1]{\langle #1\rangle}
\newcommand{\one}{\mathbf 1}
\title{Extending Tardos's Theorem to\\
Separable Convex Quadratic Programming}
\author{Bin Cui \qquad Hanqing Li \qquad Xupeng Miao\\Peking University}
\date{}
\hypersetup{pdftitle={Extending Tardos's Theorem to Separable Convex Quadratic Programming},
  pdfauthor={Bin Cui, Hanqing Li, Xupeng Miao},pdfsubject={Exact algorithms for continuous convex quadratic programming}}
\begin{document}
\maketitle

\begin{abstract}
\normalsize
For an integer equality matrix $A\in\Z^{r\times m}$ of binary encoding
length $\enc{A}$, Tardos's theorem solves linear programs using
$\poly(r,m,\enc{A})$ rational arithmetic operations and comparisons,
independently of the right-hand side and objective coefficients.
We extend this guarantee to continuous separable convex quadratic
programming over $Ax=b$, $x\ge0$: the operation count is independent
of the right-hand side, linear costs, and all nonnegative quadratic weights.
The algorithm returns an exact optimal primal--dual pair or a certificate
of infeasibility or unboundedness. All intermediate rational numbers have
encoding lengths polynomial in the total input length.

As an application, we obtain a strongly polynomial algorithm for
continuous multicommodity flow with shared capacities and quadratic
costs of aggregate arc flow, independently of demands, capacities, and
cost coefficients.

Objective compression and continuous--integer proximity yield a local
solver returning a point close to an optimum. We relax nonnegativity on
coordinates with zero optimal dual slack, minimize the largest remaining
coordinate, and optimize over the relaxed coordinates. A descent LP over
the recession cone provides a complementary bound. Circuit decomposition
combines these bounds into a scale ensuring that each local call reveals
a new coordinate. At most $m$ calls, followed by two LPs, recover an exact
optimum and its dual certificate.
\end{abstract}
\clearpage

\section{Introduction}
For an integer matrix $A\in\Z^{r\times m}$ and rational vectors $b,c$,
Tardos's theorem~\cite{Tardos1986} solves the linear program
$\min\{c^{\top}x:Ax=b,\ x\ge0\}$ with
$\poly(r,m,\enc{A})$ basic operations, independently of the encoding
lengths of $b,c$. Here $\enc{A}$ denotes the binary encoding length
of $A$, and basic operations are addition, subtraction, multiplication,
division, and comparison of rational numbers.
This yields strongly polynomial algorithms for linear programs with
small integer constraint coefficients. We extend this guarantee to
separable convex quadratic objectives.

We consider the continuous problem
\begin{equation}\label{eq:problem}
 \min\{\Phi(x):Ax=b,\ x\ge0\},\qquad
 \Phi(x)=\tfrac12\sum_{e=1}^m q_ex_e^2+\sum_{e=1}^m c_ex_e,
\end{equation}
where $A\in\Z^{r\times m}$, $b\in\Q^r$, $c\in\Q^m$, and
$q\in\Q^m_{\ge0}$. Write $Q=\operatorname{diag}(q)$, so that
$\Phi(x)=\tfrac12x^{\top}Qx+c^{\top}x$. The quadratic weights may be zero.

V\'egh~\cite{Vegh2016} obtained a strongly polynomial algorithm for separable
convex quadratic network flows. We give a Tardos-type extension
to arbitrary integer equality matrices.

\begin{theorem}\label{thm:main}
Problem~\eqref{eq:problem} can be solved with
$\poly(r,m,\enc{A})$ basic operations,
independently of the encoding lengths of $b,c,q$. The algorithm returns
an exact optimal primal--dual pair, a certificate of infeasibility, or a
feasible point and a direction certifying unboundedness. Every
intermediate rational number has encoding length polynomial in the
complete input.
\end{theorem}

For matrix classes whose entry
encodings are bounded by a fixed polynomial in the dimensions, the
theorem gives a strongly polynomial algorithm; in particular, this
holds for $\{0,\pm1\}$ matrices.
We refer to $A$ and the dimensions as the structural input.

\paragraph{Multicommodity congestion.}
Continuous multicommodity flow with shared capacities and quadratic
costs of aggregate arc flow admits a simple reformulation: introducing
aggregate-flow variables and capacity slacks gives a separable convex
quadratic program with a $\{0,\pm1\}$ equality matrix.
Theorem~\ref{thm:main} therefore yields a strongly polynomial algorithm
(Corollary~\ref{cor:flow}).

\paragraph{Organization.}
Section~\ref{sec:structure} develops the structural tools and
preprocessing. Section~\ref{sec:local} constructs the local solver,
and Section~\ref{sec:algorithm} gives the revelation algorithm and
proves Theorem~\ref{thm:main}. Section~\ref{sec:applications} establishes
the multicommodity-flow consequence. Section~\ref{sec:discussion} discusses
directions for further work.

\subsection{Technical overview}
Our algorithm addresses two questions: how to compute a point
close to a continuous optimum independently of the objective
coefficients, and how to use such points to recover an exact optimum
in a number of calls bounded by the dimensions.

\paragraph{A compact local problem.}
Suppose a box of radius $B\varepsilon$ around a feasible point contains an
optimum. We construct a feasible point within $\ell_\infty$ distance
$\varepsilon$ of an optimum. Translation and scaling give a homogeneous equality system,
and inward rounding of the bounds preserves the integer feasible points.
Frank--Tardos compression~\cite{FrankTardos1987} applied to the quadratic
and linear coefficients preserves objective comparisons on these
integer points. We then solve the compressed continuous quadratic
program. Two-sided continuous--integer proximity connects its solution
to a continuous optimum for the uncompressed objective through a shared
integer optimum.
The compressed problem has encoding length polynomial in $\log B$
and the structural input, giving the required operation bound for
the local computation. The integer optima are used only in the proof.

\paragraph{From proximity to recovery.}
The algorithm maintains a set $F$ of coordinates with zero optimal
dual slack and relaxes their nonnegativity constraints. Its goal is
to reveal a new coordinate in $J=\bar F$. A local point $x$ within
distance $\varepsilon$ of an optimum certifies zero optimal dual slack
on any coordinate $j\in J$ with $x_j>\varepsilon$.
To obtain a suitable center, the algorithm first solves an LP that
minimizes the largest coordinate in $J$, then fixes these coordinates
and optimizes over $F$ to obtain a feasible trial point $t$.
The first step makes the external
coordinates small; the second gives stationarity along feasible
directions supported in $F$.
A descent LP over the recession cone measures first-order descent
along directions that are nonnegative on $J$, with bounds $1/q_e$ on
the positive-curvature coordinates of a direction.
The quadratic weights enter only the objective and the right-hand sides
of the bounds. The two LP values define a scale $E_F$.

Conformal circuit decomposition is the main tool connecting these
auxiliary problems to the local solver. It decomposes displacements
into integer circuits of structurally bounded size without coordinate
cancellation. For the covering bound, decompose the displacement from
$t$ to a nearest optimum. A circuit that decreases an external
coordinate has its coefficient bounded by that coordinate of $t$;
a circuit with nonnegative external coordinates and positive curvature
is controlled by the descent LP. Choosing a nearest optimum excludes
redundant displacements along flat directions. This places an optimum
within a structural multiple of $E_F$ of $t$.
For the progress bound, stationarity on $F$ ensures that circuits
supported there contribute only to zero-curvature coordinates.
The remaining circuits bound $E_F$ by a structural multiple of the
largest external coordinate of every optimum.
These bounds prescribe a local box and an accuracy that guarantee a new
tight coordinate on every call. At most $m$ local calls suffice.
At zero scale, the trial point is optimal for the relaxed problem,
so its positive-curvature coordinates equal those of an original optimum.
Fixing these coordinates leaves an LP in the zero-curvature variables;
a second LP computes a dual certificate. Each iteration reconstructs
its numerical data from the original input and $F$, giving polynomial
intermediate encodings.

\subsection{Related work}
Granot and Skorin-Kapov~\cite{GranotSkorin1990} extended Tardos's
framework to strictly convex quadratic programming. Their operation
bound is independent of $b,c$ but depends on the numerical entries of
both the constraint and quadratic matrices.
For network flow constraints, V\'egh~\cite{Vegh2016} gave a strongly
polynomial algorithm for separable convex quadratic costs,
independent of demands, capacities,
and all cost coefficients. He identified a Tardos-type extension
beyond network constraints as a direction for further work.

Compression of integer quadratic and polynomial objectives was studied
by Granot and Skorin-Kapov~\cite{GranotSkorin1989} and Etscheid,
Kratsch, Mnich, and R\"oglin~\cite{EtscheidEtAl2017}. Hochbaum and
Shanthikumar~\cite[Theorem~3.3]{HochbaumShanthikumar1990} proved
two-sided continuous--integer proximity for separable convex objectives;
see also~\cite{GranotSkorinProximity1990}. The transfer between
continuous problems through a common integer optimum appears
in~\cite[Theorem~3.8]{HochbaumShanthikumar1990}.
The outer algorithm builds on tight-coordinate revelation and linear
recovery in V\'egh's framework~\cite{Vegh2016}.

Karzanov and McCormick give polynomial algorithms for separable convex
optimization on unimodular linear spaces~\cite{KarzanovMcCormick1997}.
Klimm and Warode study the parametric computation of flows with
aggregate multicommodity costs using convex-programming
subproblems~\cite{KlimmWarode2022}. Our multicommodity result gives an
exact, strongly polynomial algorithm for a single instance.

\section{Structural tools and preprocessing}\label{sec:structure}
We first record the algorithmic interfaces, then establish the circuit
and exchange tools and reduce to the feasible, finite case.
\begin{enumerate}[label=(\roman*)]
\item An LP with an integer constraint matrix $D$ can be solved exactly with
$\poly(\text{dimensions},\enc{D})$ basic operations, independently of
the objective and right-hand side, with polynomial intermediate encodings
\cite{Tardos1986}.
In the feasible finite case, we obtain an optimal primal--dual pair.
Tardos's bound is polynomial in the dimensions and $\log\Delta(D)$,
where $\Delta(D)$ is the maximum of one and the absolute values of all
square subdeterminants of $D$. Hadamard's inequality gives the stated
encoding-length bound.
This guarantee also applies to LPs with variable bounds and free
variables, using unit coefficients for bounds and signed column copies
for free variables.
\item For $w\in\Q^d$ and an integer $T\ge2$, Frank--Tardos compression
returns $\bar w\in\Z^d$ in $\poly(d,\log T)$ basic operations such that
\begin{equation}\label{eq:FT}
 \operatorname{sign}(w^{\top}a)=\operatorname{sign}(\bar w^{\top}a)
 \quad(a\in\Z^d,\ \norm{a}_1\le T-1).
\end{equation}
One may take $\norm{\bar w}_\infty\le2^{4d^3}T^{d(d+2)}$, so its
encoding is $\poly(d,\log T)$. Intermediate encodings are polynomial
in the input~\cite{FrankTardos1987}.
\item A rational convex quadratic program can be solved exactly in time
polynomial in its complete binary encoding~\cite{Kozlov1980}.
In the feasible finite case, it returns a rational optimal point of
polynomial encoding length.
We use this algorithm on bounded instances with structurally bounded encoding.
Their integer data can be converted to binary by bounded
doubling and comparison. Simulating its bit algorithm and reconstructing
its rational output then takes polynomially many basic operations
in this structural encoding length.
\end{enumerate}

\subsection{Circuits}
We retain an independent subset of the original rows of $A$ and verify
consistency of $b$ with the deleted rows. Thus, unless infeasibility has
already been certified, $A$ has full row rank. If the rank is zero, the
problem separates into one-dimensional quadratics and linear functions.
Assume henceforth $r\ge1$. Define
\begin{equation}\label{eq:G}
 U=\max(1,\max_{i,e}|A_{ie}|),\qquad G=(rU)^r.
\end{equation}
A circuit of a matrix is a nonzero kernel vector with inclusion-minimal
support. A circuit of an integer matrix has a primitive integer representative.
By the cofactor construction on at most $r+1$ columns and Hadamard's
inequality, every circuit of $A$ has a representative $h$ satisfying
\begin{equation}\label{eq:circuitbound}
 1\le |h_e|\le G\quad(e\in\supp h).
\end{equation}
In particular, $\log G=O(r\log(rU))$.

Two vectors are sign-compatible if their coordinate products are nonnegative.
A vector $h$ is conformal to $v$ if it is sign-compatible with $v$ and its
support is contained in that of $v$.

We use the conformal decomposition of a linear subspace into elementary
vectors~\cite[Theorem~1]{Rockafellar1969}.

\begin{lemma}[Conformal decomposition]\label{lem:decomp}
Every nonzero kernel vector $v$ of a matrix $D$ admits a representation
$v=\sum_{i=1}^k\lambda_i h^i$, where $\lambda_i>0$, each $h^i$ is a
circuit conformal to $v$, and $k\le\dim\ker D$. The sum has no coordinate
cancellation. If $D$ is integer, the circuits can be taken primitive integer.
\end{lemma}
\begin{proof}
Intersect $\ker D$ with the closed orthant and coordinate support of $v$.
This is a pointed cone whose extreme rays are circuits: a vector with a
non-minimal support admits a nontrivial decomposition within that orthant
by perturbing along a smaller-support kernel vector. Subtracting the
largest feasible multiple of an extreme ray eliminates at least one
nonzero coordinate. Repeating gives a finite conformal decomposition.

If its rays are linearly dependent, choose a relation
$\sum_i\alpha_i h^i=0$. The coefficients have both signs, because all
rays lie in the same pointed orthant. Subtract
$\min_{\alpha_i>0}(\lambda_i/\alpha_i)$ times this relation from the
coefficient vector. Coefficients remain nonnegative and at least one
vanishes. Repeat until the rays are independent. Their number is then
at most $\dim\ker D$.
\end{proof}

\subsection{Feasibility, unboundedness, and attainment}
Let $N=\{e:q_e>0\}$ and $Z=\{e:q_e=0\}$.
First solve the feasibility LP $Ax=b,x\ge0$. Infeasibility has a
Farkas certificate $A^{\top}y\ge0$, $b^{\top}y<0$.
If $N=\varnothing$, use the LP solver for the entire problem.
Otherwise test feasibility of
\begin{equation}\label{eq:lowerpotential}
 A_Z^{\top}\pi_0\le c_Z,\qquad \pi_0\text{ free}.
\end{equation}
If this system is infeasible, its Farkas certificate gives
$v_Z\ge0$, $A_Zv_Z=0$, $c_Z^{\top}v_Z<0$.
Extending it by $v_N=0$ certifies unboundedness of the QP from the
already computed feasible point. Both LP matrices contain only the
original integer matrix or its transpose, signed copies, and unit coefficients.

Suppose that~\eqref{eq:lowerpotential} is feasible. We show that the
objective is bounded below and attains its minimum.
Every feasible point satisfies
\begin{equation}\label{eq:coercive}
 \Phi(x)\ge\tfrac12\sum_{e\in N}q_ex_e^2+
 (c_N-A_N^{\top}\pi_0)^{\top}x_N+\pi_0^{\top}b.
\end{equation}
This is bounded below and coercive in $x_N$.
For fixed $u=x_N$, the linear-fiber value
\[
 v(u)=\min\{c_Z^{\top}z:A_Zz=b-A_Nu,\ z\ge0\}
\]
has a closed polyhedral epigraph and a closed feasible domain, by
polyhedral projection. It attains its value whenever feasible, with
the affine lower bound given by $\pi_0$. Thus
$\frac12u^{\top}Q_Nu+c_N^{\top}u+v(u)$, restricted to $u\ge0$, has a nonempty
compact sublevel set and attains its minimum.
A minimizing linear fiber proves attainment of
the QP. Thus, for a feasible QP, the objective is unbounded below if and
only if~\eqref{eq:lowerpotential} is infeasible.

If $m=r$, the equalities determine a unique feasible point. Rational
elimination returns it, with potential $\pi=(A^{\top})^{-1}(Qx+c)$.
Henceforth the instance is feasible and finite, and
\begin{equation}\label{eq:counts}
 1\le r<m,\qquad N\ne\varnothing.
\end{equation}
We use the following necessary and sufficient KKT conditions to certify
optimality and justify the nonnegativity relaxations in
Section~\ref{sec:algorithm}.
\begin{equation}\label{eq:KKT}
 Ax=b,\quad x\ge0,\quad \sigma=Qx+c-A^{\top}\pi\ge0,\quad x^{\top}\sigma=0.
\end{equation}
These multipliers follow from the normal-cone formula for a polyhedron.

Certificates are reported for the original equality system.
Multipliers for the retained rows are extended by zero on the deleted
rows. If a deleted equation is inconsistent with the retained equations,
its row-dependence relation gives a Farkas certificate, with the sign
chosen so that its inner product with the original right-hand side is
negative.

\subsection{Separable exchange}
We use the following exchange inequality to prove the two-sided
continuous--integer proximity bound in Section~\ref{sec:local}.

\begin{lemma}[Separable exchange]\label{lem:exchange}
Let $\psi(u)=\sum_e\psi_e(u_e)$ be separable convex. If $h$ is conformal
to $w-v$ and $|h_e|\le|w_e-v_e|$, then
\begin{equation}\label{eq:exchange}
 \psi(v+h)+\psi(w-h)\le\psi(v)+\psi(w).
\end{equation}
\end{lemma}
\begin{proof}
For $w_e\ne v_e$ put $\theta_e=h_e/(w_e-v_e)\in[0,1]$.
The two exchanged coordinates are complementary convex combinations
of $v_e,w_e$. Add their convexity inequalities and sum over $e$.
When $w_e=v_e$, also $h_e=0$.
\end{proof}
If $v$ and $w$ are optimal in their respective feasible sets and the
exchanged points remain feasible in the corresponding sets, both
objective changes are nonnegative. The exchange inequality makes their
sum nonpositive, so both changes are zero and the exchanged points
are optimal.

\section{A local solver via objective compression}\label{sec:local}
For any index set $F$, write $P_F(b)=\{x:Ax=b,\ x_{\bar F}\ge0\}$.
The local solver finds a point close to an optimum of the relaxed problem.
Let $a\in P_F(b)\cap\Q^m$, $\varepsilon\in\Q_{>0}$, and let $B\ge1$ be an integer.
Assume the box $|x-a|\le B\varepsilon\one$ contains an optimum on $P_F(b)$.

\begin{proposition}[Local near-point solver]\label{prop:local}
There is an algorithm that returns $x\in P_F(b)$ satisfying
\begin{equation}\label{eq:local}
 \norm{x-z}_\infty\le\varepsilon,
\end{equation}
for some optimum $z$ on $P_F(b)$. It uses
$\poly(r,m,\enc{A},\log B)$ basic operations,
independently of the encoding lengths of $a,\varepsilon,q,c$.
\end{proposition}

The construction of this local solver was discovered by GPT-6.1 Sol,
as described in the AI methodology statement.

\paragraph{Normalization and rounding.}
Set $M=2mG$, $\delta=\varepsilon/M$, and $R=MB$, so that
$R\delta=B\varepsilon$ and $\log R$ is polynomially bounded in the
structural input and $\log B$.
Write $u=(x-a)/\delta$. Then $Au=0$, $|u|\le R\one$, and for
$e\notin F$ also $u_e\ge-a_e/\delta$.
Let $\ell_e$ be the maximum of $-R$ and this lower bound for $e\notin F$,
and set $\ell_e=-R$ on $F$. Each $\ell_e\le0$. Replace $\ell_e$ by its
ceiling $\bar\ell_e$, which changes it by less than one and keeps $0$
feasible. The rounding step takes $O(m\log(R+1))$ basic operations by binary
search on the integer interval $[-R,0]$.

\paragraph{Continuous--integer proximity.}
We use the following box version of the two-sided proximity
theorem~\cite[Theorem~3.3]{HochbaumShanthikumar1990}.

\begin{lemma}[Two-sided integer proximity]\label{lem:integer}
Let $P=\{u:Au=\beta,\ \ell\le u\le h\}$ be bounded and let
$P\cap\Z^m\ne\varnothing$. The bounds need not be integral.
For any separable convex quadratic $\psi$, every continuous optimum
has an integer optimum within $\ell_\infty$ distance $mG$, and
every integer optimum has a continuous optimum within the same distance.
\end{lemma}
\begin{proof}
Fix one of the two kinds of optimum, and select a nearest $\ell_1$
optimum of the other kind. The finite integer feasible set and compact
continuous optimal set ensure existence. Write the continuous point as
$v$ and the integer point as $w$. Decompose
$w-v=\sum_{i=1}^k\lambda_i h^i$ conformally into primitive integer
circuits, with $k\le m$. If $\lambda_i\ge1$, put $h=h^i$.
Both $v+h$ and $w-h$ lie coordinatewise between $v,w$, so they satisfy
the same equalities and box bounds; $w-h$ remains integer.
Lemma~\ref{lem:exchange} and continuous and integer optimality show
that both exchanged points are optimal in their respective problems.
If $v$ was fixed, $w-h$ is a closer integer optimum; if $w$ was fixed,
$v+h$ is a closer continuous optimum. Both contradict the selection.
Thus every $\lambda_i<1$, giving
$\norm{w-v}_\infty\le\sum_i\lambda_i\norm{h^i}_\infty\le mG$.
\end{proof}

\paragraph{Compression and the local algorithm.}
Let $P$ denote the scaled box before rounding, and $\bar P$ the box with
lower bounds $\bar\ell=\lceil\ell\rceil$. Rounding preserves integer
feasible points, and both boxes contain zero. Hence
\begin{equation}\label{eq:sameinteger}
 P\cap\Z^m=\bar P\cap\Z^m\ne\varnothing.
\end{equation}
After substituting $x=a+\delta u$ and subtracting the constant $\Phi(a)$,
the objective has coefficient vector
\begin{equation}\label{eq:features}
 \theta=\left((q_e\delta^2/2)_{e=1}^m,
              (\delta(q_ea_e+c_e))_{e=1}^m\right)
\end{equation}
against features $((u_e^2)_{e=1}^m,(u_e)_{e=1}^m)$.
The compression dimension is $2m$. For integer $u,v\in[-R,R]^m$, their
feature difference has $\ell_1$ norm at most $mR^2+2mR$.
Apply Frank--Tardos with
\begin{equation}\label{eq:FTparameter}
 T=1+mR^2+2mR.
\end{equation}
This preserves all comparisons and ties on the integer box, and therefore
the integer optimal set on~\eqref{eq:sameinteger}.
Unit-vector tests in~\eqref{eq:FT} preserve the sign of every quadratic
coefficient, including zeros. The compressed objective is therefore
still separable and convex.

Solve the compressed continuous QP on $\bar P$, and let its
solution be $\bar u$. Its complete encoding is polynomial in
$r,m,\enc{A},\log R$: its demand is zero, its bounds are integers of
magnitude at most $R$, and its objective is compressed. Return
$x=a+\delta\bar u$.

\begin{proof}[Proof of Proposition~\ref{prop:local}]
Lemma~\ref{lem:integer} puts $\bar u$ within $mG$ of a compressed
integer optimum $w$. Comparison and tie preservation, together with
\eqref{eq:sameinteger}, make $w$ an integer optimum for the uncompressed
objective on the \emph{unrounded} box $P$. Apply the reverse direction
of the lemma to this objective on $P$: $w$ is within $mG$ of a continuous
optimum $u^*$ there. The unrounded box contains a global optimum of the
current relaxed problem, so $a+\delta u^*$ is such a global optimum.
Moreover, $\bar P\subseteq P$, ensuring feasibility of $x$.
The triangle inequality and rescaling give distance at most
$2mG\delta=\varepsilon$, proving~\eqref{eq:local}.
The standard interfaces and compressed encoding give the operation count.
\end{proof}

\section{Coordinate revelation and recovery}\label{sec:algorithm}
We use the local solver to reveal a new tight coordinate on each call,
then recover an optimal primal--dual pair by linear programming.
For the analysis, fix an optimal primal--dual pair $(x^*,\pi^*)$ of the
original problem. The algorithm does not need to know this pair.
Put
$\sigma^*=Qx^*+c-A^{\top}\pi^*\ge0$ and $T^*=\{e:\sigma_e^*=0\}$.
The algorithm maintains a growing set $F\subseteq T^*$ and removes
the nonnegativity constraints on its coordinates. It therefore works on
$P_F(b)=\{x:Ax=b,\ x_{\bar F}\ge0\}$. For every $x\in P_F(b)$,
\begin{equation}\label{eq:safeidentity}
 \Phi(x)-\Phi(x^*)=
 \tfrac12\sum_{e\in N}q_e(x_e-x_e^*)^2+
 \sum_{e\notin F}\sigma_e^*x_e\ge0.
\end{equation}
Thus the relaxed minimum equals the original minimum. Every relaxed
optimum $z$ has $z_N=x_N^*$, and $z_e>0$ outside $F$ implies $e\in T^*$.
Write
\[
 \kappa=mG,\qquad \rho_F(z)=\max_{e\notin F}z_e,
\]
with an empty maximum taken as zero.

\subsection{Small external coordinates and stationarity on \texorpdfstring{$F$}{F}}\label{sec:trial}
We construct a trial point with small external coordinates that is
optimal on the affine fiber obtained by fixing them.
First minimize the largest external coordinate:
\begin{equation}\label{eq:smalloutside}
\begin{aligned}
 \alpha_F=\min_{u,\alpha}\quad&\alpha\\
 \text{subject to}\quad&Au=b,\qquad
 0\le u_{\bar F}\le\alpha\one,\qquad \alpha\ge0,\\
 &u_F\text{ free}.
\end{aligned}
\end{equation}
This LP is feasible and has a finite minimum. Its constraint matrix
contains only $A$ and entries in $\{0,\pm1\}$. Choose a basic optimal
solution after standardization, and let $v=u_{\bar F}$. For every
$z\in P_F(b)$,
\begin{equation}\label{eq:alphabound}
 0\le v\le\alpha_F\one,\qquad \alpha_F\le\rho_F(z),
\end{equation}
since $(z,\rho_F(z))$ is feasible in~\eqref{eq:smalloutside}.
In particular, $\alpha_F=0$ if $\bar F$ is empty.

Fix the external coordinates and minimize over the free coordinates
by solving the linear system
\begin{equation}\label{eq:trial}
 t_{\bar F}=v,\qquad A_Ft_F=b-A_{\bar F}v,\qquad
 Q_Ft_F+c_F=A_F^{\top}\lambda.
\end{equation}
We use rational elimination, setting free unknowns to zero.
To prove consistency, put $\beta=b-A_{\bar F}v\in\im A_F$.
On the affine fiber $A_Fw=\beta$, the objective
$\Phi_F(w)=\tfrac12w^{\top}Q_Fw+c_F^{\top}w$ equals
\[
 \tfrac12\norm{Q_F^{1/2}(w-x_F^*)}_2^2+
 (\pi^*)^{\top}\beta-\tfrac12(x_F^*)^{\top}Q_Fx_F^*,
\]
because $F\subseteq T^*$. The image of the nonempty affine fiber under
$w\mapsto Q_F^{1/2}(w-x_F^*)$ is an affine subspace.
Its minimum-norm point exists, so the displayed objective attains its
minimum.
Stationarity gives~\eqref{eq:trial}. Thus $t\in P_F(b)$ and
$0\le t_{\bar F}\le\alpha_F\one$.

\subsection{Descent over the recession cone}\label{sec:descent}
Put $g=Qt+c$ and compute the optimal value of the following descent LP
over the recession cone:
\begin{equation}\label{eq:descentLP}
\begin{aligned}
 V_F=\max_d\quad&-g^{\top}d\\
 \text{subject to}\quad&Ad=0,\qquad d_{\bar F}\ge0,\\
 &-1/q_e\le d_e\le1/q_e\qquad(e\in N).
\end{aligned}
\end{equation}
The bounds $|q_ed_e|\le1$ on $N$ allow us to bound the descent value
by an unweighted $\ell_1$ distance to an optimum on these coordinates.
Coordinates in $Z$ satisfy only the equality and sign constraints.
The constraint matrix contains only $A$ and unit coefficients;
$q$ enters the right-hand side and the objective.

Every feasible $d$ is a recession direction of $P_F(b)$. Hence for
any relaxed optimum $z$, first-order optimality gives
$\nabla\Phi(z)^{\top}d\ge0$, and
\[
 -g^{\top}d\le (Q(z-t))^{\top}d\le\sum_{e\in N}|z_e-t_e|.
\]
Since $d=0$ is feasible, the LP attains a finite optimum satisfying
\begin{equation}\label{eq:descentbound}
 0\le V_F\le\norm{(z-t)_N}_1
 \qquad\text{for every relaxed optimum }z.
\end{equation}
Define the scale
\begin{equation}\label{eq:scale}
 E_F=\max(\alpha_F,V_F).
\end{equation}

\subsection{Coverage and progress}\label{sec:scale}
The same scale $E_F$ provides both a box containing an optimum and
a guarantee that each local call reveals a new coordinate.

\begin{lemma}[Covering bound]\label{lem:cover}
There is an optimum $z_0$ on $P_F(b)$ with
\begin{equation}\label{eq:cover}
 \norm{z_0-t}_\infty\le\kappa E_F.
\end{equation}
\end{lemma}
\begin{proof}
Choose $z_0$ minimizing $\norm{z_0-t}_1$ in the relaxed optimal set.
This choice exists because the set is nonempty and closed and distance
sublevel sets are compact. If $z_0=t$, the bound is immediate.
Otherwise, Lemma~\ref{lem:decomp} gives a conformal decomposition
\[
 w=z_0-t=\sum_{i=1}^k\lambda_i h^i,\qquad k\le m,
\]
with primitive integer circuits $h^i$ satisfying
$\norm{h^i}_\infty\le G$. We show that each $\lambda_i\le E_F$.
Fix a term $\lambda h$.

If $h_e<0$ for some $e\notin F$, conformality gives
$\lambda|h_e|\le t_e-z_{0,e}\le\alpha_F$.
Since $|h_e|\ge1$, we have $\lambda\le\alpha_F$.
Otherwise $h_{\bar F}\ge0$. A sufficiently small step from $z_0$
along $-h$ stays coordinatewise between $z_0$ and $t$, so it is feasible.
Optimality and conformality imply
\begin{equation}\label{eq:circuitdescent}
 -g^{\top}h\ge h^{\top}Qw\ge\lambda h^{\top}Qh.
\end{equation}
If $s=h^{\top}Qh>0$, integer coordinates give
$q_e|h_e|\le q_eh_e^2\le s$ for every $e$.
Thus $h/s$ is feasible in~\eqref{eq:descentLP}, and
\eqref{eq:circuitdescent} yields $V_F\ge\lambda$.
If $s=0$, then $Qh=0$. Since $h$ is a recession direction and the
relaxed minimum is finite, $c^{\top}h\ge0$; optimality along $-h$ gives
$c^{\top}h\le0$. The objective is therefore constant along $h$.
A small step from $z_0$ along $-h$ would be optimal and strictly closer
to $t$, a contradiction. We conclude that all $\lambda_i\le E_F$,
and hence $\norm{z_0-t}_\infty\le mGE_F=\kappa E_F$.

\end{proof}

\begin{lemma}[Progress bound]\label{lem:progress}
Every optimum $z$ on $P_F(b)$ satisfies
\begin{equation}\label{eq:external}
 E_F\le m\kappa\,\rho_F(z).
\end{equation}
\end{lemma}
\begin{proof}
Fix any relaxed optimum $z$. By Lemma~\ref{lem:decomp}, write a conformal
decomposition
\[
 z-t=\sum_{i=1}^k\lambda_i h^i,\qquad \lambda_i>0,\qquad k\le m,
\]
where the $h^i$ are primitive integer circuits with
$\norm{h^i}_\infty\le G$, using the empty sum if $z=t$.
If a circuit $h^i$ is supported in $F$,
both $h^i$ and $-h^i$ are feasible directions within the respective
affine fibers through $t$ and $z$. Stationarity at $t$ and optimality
of $z$ for the relaxed problem give
\[
 \nabla\Phi(t)^{\top}h^i=\nabla\Phi(z)^{\top}h^i=0.
\]
Subtracting and using conformality, we obtain
\[
 0=(h^i)^{\top}Q(z-t)\ge\lambda_i(h^i)^{\top}Qh^i\ge0.
\]
Thus $Qh^i=0$, so this circuit contributes nothing on $N$.
Every other circuit has a nonzero coordinate $e\notin F$.
By~\eqref{eq:alphabound}, both $z_e$ and $t_e$ lie in
$[0,\rho_F(z)]$. Conformality and integer coordinates give
$\lambda_i\le\lambda_i|h^i_e|\le|z_e-t_e|\le\rho_F(z)$.
Summing the contributions on $N$ gives
\[
 \norm{(z-t)_N}_\infty\le mG\rho_F(z)=\kappa\rho_F(z).
\]
Now~\eqref{eq:descentbound} gives $V_F\le m\kappa\rho_F(z)$.
Together with $\alpha_F\le\rho_F(z)$, this proves~\eqref{eq:external}.
\end{proof}

\subsection{Revelation and recovery}\label{sec:recovery}
Starting with $F=\varnothing$, repeat the following steps:
\begin{enumerate}[label=\arabic*.]
\item Solve~\eqref{eq:smalloutside}, construct $t$ by~\eqref{eq:trial},
and solve~\eqref{eq:descentLP}. Set $E_F=\max(\alpha_F,V_F)$.
\item If $E_F=0$, perform the recovery below and return.
\item Set $\varepsilon=E_F/(4m\kappa)$ and $B=4m\kappa^2$.
Apply Proposition~\ref{prop:local} with center $t$, accuracy
$\varepsilon$, and radius $B\varepsilon=\kappa E_F$.
\item For the returned point $x$, update
\begin{equation}\label{eq:reveal}
 F\ \leftarrow\ F\cup\{e\notin F:x_e>\varepsilon\}.
\end{equation}
\end{enumerate}

Lemma~\ref{lem:cover} ensures that every local box contains a relaxed
optimum. The local solver returns $x$ within distance $\varepsilon$
of some relaxed optimum $z$, which satisfies, by Lemma~\ref{lem:progress},
\[
 \rho_F(z)\ge E_F/(m\kappa)=4\varepsilon.
\]
Thus some external coordinate has $x_e\ge3\varepsilon>\varepsilon$
and is added. Every added coordinate $j$ has
$z_j\ge x_j-\varepsilon>0$, so~\eqref{eq:safeidentity} gives
$\sigma_j^*=0$. The invariant is preserved and $F$ strictly grows.
There are at most $m$ local calls. If $F$ contains every coordinate,
\eqref{eq:external} gives $E_F=0$.

\paragraph{Recovery by linear programming.}
If $E_F=0$, then $\alpha_F=V_F=0$ and $t_{\bar F}=0$.
For any $u\in P_F(b)$, the displacement $u-t$ satisfies
$A(u-t)=0$ and $(u-t)_{\bar F}\ge0$. A sufficiently small positive
multiple is feasible in~\eqref{eq:descentLP}, so $V_F=0$ gives
$g^{\top}(u-t)\ge0$. Convexity implies that $t$ is optimal for the
relaxed problem.
By~\eqref{eq:safeidentity}, $t_N=x_N^*\ge0$.

Fix $x_N=t_N$ and solve the LP
\begin{equation}\label{eq:recovery}
 \min\{c_Z^{\top}x_Z:A_Zx_Z=b-A_Nt_N,\ x_Z\ge0\}.
\end{equation}
Fixing these coordinates retains the original optimum $x^*$, and
the remaining objective is linear. Thus this LP is feasible and finite,
and any optimum, combined with $x_N=t_N$, is an original optimum $x$.
If $Z=\varnothing$, take $x=t$.

Finally, compute an optimality certificate by solving the linear
feasibility problem in the free variable $\pi$:
\begin{equation}\label{eq:certificate}
 A^{\top}\pi\le Qx+c,\qquad
 (A^{\top}\pi)_e=(Qx+c)_e\quad\text{for every }e\text{ with }x_e>0.
\end{equation}
The KKT conditions~\eqref{eq:KKT} at the optimum $x$ guarantee
feasibility. Every solution $\pi$ gives
$\sigma=Qx+c-A^{\top}\pi\ge0$ with $x^{\top}\sigma=0$, and hence forms an
optimal primal--dual pair with $x$.

\subsection{Complexity}
Each LP has polynomial dimensions and an integer constraint matrix
built from $A$, its transpose, signed copies, and unit coefficients.
In particular, the recovery LP uses $A_Z$, and the certificate LP uses
rows of $A^{\top}$; $t_N$ and $Qx+c$ occur only on their right-hand sides.
The trial system
uses polynomially many rational elimination operations. Since
$\log\kappa=O(\log m+r\log(rU))$, the local parameter
$B=4m\kappa^2$ has structurally bounded encoding.
Proposition~\ref{prop:local} therefore uses
$\poly(r,m,\enc{A})$ basic operations per call. There are at most $m$ local
calls, at most $m+1$ iterations of the two LPs and the trial system,
and at most two LPs for recovery and certification.
This proves the operation bound.

Let $\mathcal L$ be the complete input length. For every possible $F$,
the LP~\eqref{eq:smalloutside} is built directly from the original input.
A basic optimal solution after standardization has encoding length
$\poly(\mathcal L)$, uniformly over $F$. Thus $\alpha_F,v$, and the
right-hand side of~\eqref{eq:trial} have polynomial encodings.
Determinant bounds for rational elimination give the same guarantee
for $t$. The descent LP has numerical data $Qt+c$
and $1/q_e$ of polynomial encoding length. Choose basic optimal
solutions after standardization; their encodings, and those of
$V_F,E_F,\varepsilon$, are polynomial in $\mathcal L$.
The LP and elimination algorithms also control intermediate numbers.

Inside the local solver, $\delta=\varepsilon/(2mG)$ and the compression
input have polynomial encodings. Frank--Tardos controls its
intermediate numbers, while the local QP has structurally bounded
encoding. Its rational output and the rescaled point have polynomial
encodings. Only the index set $F$ enters the next iteration; all
numerical data are reconstructed from the original input and $F$.
The recovery LP has polynomially encoded right-hand side
$b-A_Nt_N$ and costs $c_Z$, so a basic optimal solution has polynomial
encoding length. The certificate LP then has right-hand side
$Qx+c=Qt+c$; the support of $x$ only selects its equality rows.
A basic feasible solution after standardization has polynomial
encoding length, and both LP computations obey the same intermediate
encoding bound.
This completes the proof of Theorem~\ref{thm:main}.

\section{Continuous multicommodity congestion}\label{sec:applications}
Let $(V,E)$ be a directed graph with node--arc incidence matrix $D$.
For commodities $k=1,\ldots,p$, let $x^k\ge0$ satisfy $Dx^k=b^k$.
A subset $E_{\rm cap}\subseteq E$ has shared finite capacities
$\sum_{k=1}^p x_e^k\le u_e$.
Consider the objective
\begin{equation}\label{eq:flowcost}
 \sum_{e\in E}\left[
 \tfrac12\alpha_e\left(\sum_{k=1}^p x_e^k\right)^2
 +\beta_e\sum_{k=1}^p x_e^k\right]
 +\sum_{k=1}^p\sum_{e\in E}c_e^k x_e^k,\qquad \alpha_e\ge0,
\end{equation}
with rational data throughout.

\begin{corollary}\label{cor:flow}
Continuous multicommodity flow with shared capacities and
objective~\eqref{eq:flowcost} has a strongly polynomial exact algorithm.
The operation count is polynomial in the size of the explicit
commodity--arc formulation, independently of the demands, capacities,
and all cost coefficients. The algorithm returns an optimum with an
optimality certificate, an infeasibility certificate, or a feasible point
and an unboundedness direction. Intermediate encodings are polynomial
in the complete input. The number of commodities need not be fixed.
\end{corollary}
\begin{proof}
Introduce nonnegative aggregate variables $y$ and capacity slacks
$z\in\R^{E_{\rm cap}}_{\ge0}$, with equalities
\[
 Dx^k=b^k,\qquad
 y_e-\sum_{k=1}^p x_e^k=0\ (e\in E),\qquad
 y_e+z_e=u_e\ (e\in E_{\rm cap}).
\]
The objective becomes
\[
 \sum_{e\in E}\left(\tfrac12\alpha_e y_e^2+\beta_e y_e\right)
 +\sum_{k=1}^p\sum_{e\in E}c_e^k x_e^k.
\]
It is separable in the lifted variables: $y$ has curvature
$\alpha$, while $x^k$ and $z$ have zero curvature.
There are $p|E|+|E|+|E_{\rm cap}|$ variables and
$p|V|+|E|+|E_{\rm cap}|$ equations, all with coefficients in
$\{0,\pm1\}$. Demands and capacities occur only on the right-hand side;
costs occur only in the objective. Every original feasible flow has a
unique lift with the same cost, and every lifted feasible point projects
to such a flow. Apply Theorem~\ref{thm:main}. Its certificates verify
optimality or infeasibility through this explicit equivalence, and a
lifted unboundedness ray projects to one of the same cost.
\end{proof}

\section{Discussion}\label{sec:discussion}

A natural question is which classes of
nonseparable convex quadratic objectives admit the same independence
from objective coefficients. Our proof uses separability in the exchange
inequalities underlying proximity, in preserving convexity under
objective compression, and in the circuit inequalities of
Lemmas~\ref{lem:cover} and~\ref{lem:progress}.

Extending the guarantee to arbitrary positive semidefinite quadratic
matrices would be a more ambitious goal, as it would yield a strongly
polynomial algorithm for general linear programming.
Indeed, feasibility
of $Bx=d$, $x\ge0$ can be tested by solving
\[
  \min_{x\ge0}\tfrac12\norm{Bx-d}_2^2.
\]
The quadratic matrix is $B^{\top}B$, and the minimum is zero exactly when the
linear system is feasible. The constraints contain only nonnegativity,
so an operation bound independent of the objective coefficients would
be polynomial in the dimensions alone.

\paragraph{AI methodology.}
The local solver in Section~\ref{sec:local} was discovered by GPT-6.1 Sol.
Its construction combines Frank--Tardos objective compression with
two-sided continuous--integer proximity to obtain a point near an optimum
of the current relaxed problem from a compressed continuous quadratic
program.
The authors take full responsibility for the mathematical claims,
proofs, and references in this paper.

\end{document}